\documentclass[12pt,english]{iopart}

\usepackage[utf8]{inputenc}

\expandafter\let\csname equation*\endcsname\relax
\expandafter\let\csname endequation*\endcsname\relax

\usepackage{amsmath,graphicx,epsfig,amsthm,color,bm}
\usepackage{cite}
\eqnobysec
\usepackage{amssymb}
\usepackage{mathrsfs}
\usepackage{enumerate}
\usepackage{enumitem}
\usepackage{times,txfonts}
\usepackage{subfigure}
\usepackage{framed} 
\usepackage{ulem}
\usepackage{color}  
\usepackage[sort&compress,square,numbers]{natbib}

\newcommand{\bra}[1]{\langle #1|}
\newcommand{\ket}[1]{|#1\rangle}
\newcommand{\bracket}[2]{\langle #1 | #2 \rangle}

\newcommand{\abs}[1]{\lvert #1\rvert}

\newtheorem{theorem}{Theorem}

\newtheorem{corollary}{Corollary}

\definecolor{shadecolor}{rgb}{0.92,0.92,0.92}

\begin{document}

\title[]{No-go theorems for deterministic purification and probabilistic enhancement of coherence}

\author{Qiming Ding $^{1,2,\dagger}$}
\address{$^{1}$ \quad Center on Frontiers of Computing Studies, Department of Computer Science, Peking University, Beijing 100080, China.}
\address{$^{2}$ \quad School of Physics, Shandong University, Jinan 250100, China.}
\ead{dqiming94@gmail.com}

\author{Quancheng Liu $^{3,\dagger}$}
\address{$^{3}$ \quad Department of Physics, Institute of Nanotechnology and Advanced Materials, Bar-Ilan University, Ramat-Gan 52900, Israel}
\ead{qcliu.ac@gmail.com}

\address{$^{\dagger}$ \quad These authors contributed equally to this work.}
%\secondnote{All authors contributed equally to the work on this paper.}

\vspace{10pt}
\begin{indented}
\item[]
\end{indented}

\begin{abstract}
The manipulation of quantum coherence is one of the principal issues in the resource theory of coherence, with two critical topics being the purification and enhancement of coherence. Here, we present two no-go theorems for the deterministic purification of coherence and the probabilistic enhancement of coherence, respectively. Specifically, we prove that a quantum state cannot be deterministically purified if it can be expressed as a convex combination of an incoherent state and a coherent state. Besides, we give an easy-to-verified sufficient and necessary condition to determine whether a state can be probabilistically enhanced via a stochastic strictly incoherent operation (sSIO). Our findings provide two feasibility criteria for the deterministic purification and the probabilistic enhancement of coherence, respectively. These results have repercussions on the understanding of quantum coherence in real quantum systems.
\end{abstract}
~~~~~~~~\providecommand{\keywords}[1]{Keywords:~~#1}
~~~\keywords{quantum coherence, coherence manipulation, no-go theorem.}

%an essential component of quantum information processing, 
% Uncomment for keywords
%\vspace{2pc}
%\noindent{\it Keywords}: XXXXXX, YYYYYYYY, ZZZZZZZZZ
%
% Uncomment for Submitted to journal title message
\submitto{\JPA}
%
% Uncomment if a separate title page is required
\maketitle
% 
% For two-column output uncomment the next line and choose [10pt] rather than [12pt] in the \documentclass declaration
%\ioptwocol

\section{Introduction}

Quantum coherence is a fundamental feature of quantum mechanics, describing the capability of a quantum state to exhibit quantum interference phenomena~\cite{nielsen_Quantum_2010}. It is an essential ingredient in quantum information processing and also plays a significant role in many other fields, such as quantum metrology~\cite{giovannetti_Advances_2011,demkowicz-dobrzanski_Using_2014}, nanoscale thermodynamics~\cite{cwiklinski_Limitations_2015,lostaglio_Description_2015,lostaglio_Quantum_2015,narasimhachar_Lowtemperature_2015}, and quantum biology~\cite{lloyd_Quantum_2011,huelga_Vibrations_2013,lambert_Quantum_2013}. Recently, the quantification of coherence has attracted a growing interest from the perspective of resource theory~\cite{baumgratz_Quantifying_2014,streltsov_Colloquium_2017,hu_Quantum_2018}. 

Any quantum resource theory is characterized by two fundamental ingredients, namely, free states and free operations~\cite{chitambar_Quantum_2019}. For the resource theory of coherence, the free states are quantum states which are diagonal in a prefixed reference basis. The free operations are not uniquely specified. Motivated by practical considerations, several free operations were presented, such as maximally incoherent operations (MIO)~\cite{aberg_Quantifying_2006}, incoherent operations (IO)~\cite{baumgratz_Quantifying_2014}, strictly incoherent operations (SIO)~\cite{winter_Operational_2016,yadin_Quantum_2016}, genuinely incoherent operations (GIO)~\cite{vicente_Genuine_2016}, and physical incoherent operations (PIO)~\cite{chitambar_Comparison_2016,chitambar_Critical_2016}.
%All quantum resource theories have two fundamental ingredients: free state and free operation~\cite{chitambar_Quantum_2019}. \red{For the resource
%theory of coherence,} a free state does not contain quantum coherence, \red{and is also called incoherent state}. 
%The free operation is a set of quantum operations that can never create coherence. 
%Based on different physical or mathematical considerations, 
%several \red{the free operations in the resource theory of coherence} were presented
%\red{the free operations in the resource theory of coherence are not uniquely specified,}  
With these notions, the essential conditions of coherence measures have been established~\cite{baumgratz_Quantifying_2014,streltsov_Colloquium_2017}. Based on these conditions, a number of legitimate coherence measures also have been proposed ~\cite{napoli_Robustness_2016,yu_Alternative_2016,liu_New_2017,qi_Measuring_2017,zhou_Polynomial_2017,xi_Coherence_2019,xi_Epsilonsmooth_2019,yao_Quantum_2019,yu_Quantifying_2020,li_Quantum_2021,cui_Maximalvalue_2020}, which can be used to study the role of coherence in many physical contexts  quantitatively~\cite{bromley_Frozen_2015,yu_Measureindependent_2016,streltsov_Entanglement_2016,ma_Coherence_2019,zhang_Estimating_2018,yu_Detecting_2019,ding_Efficient_2021}. Many studies have shown that the  amount of quantum coherence is directly related to the success or failure of some quantum information processing tasks, such as quantum phase discrimination~\cite{napoli_Robustness_2016}, quantum algorithm~\cite{hillery_Coherence_2016,shi_Coherence_2017}, and the secrete key rate
in quantum key distribution~\cite{ma_Operational_2019}.

Since any real quantum system inevitably interacts with the surrounding  environment. Such interactions generally spoil the coherence of the quantum states and further weaken the ability of the system to perform quantum information processing tasks. Thus, a critical issue is to investigate coherence manipulations via free operations. Many protocols of coherence manipulation have been proposed based on various physical scenarios~\cite{yuan_Intrinsic_2015,winter_Operational_2016,chitambar_Assisted_2016,fang_Probabilistic_2018,regula_OneShot_2018,chen_Oneshot_2019,torun_optimal_2019,zhao_OneShot_2019,du_Conditions_2015,bu_Catalytic_2016,bu_Maximum_2017,zhang_Oneshot_2020,liu_Catalystassisted_2020,pang_Probabilistic_2020,lami_Completing_2020,xing_Reduce_2020,liu_Optimal_2021,ding_Tightness_2021}, and some have been demonstrated in the experiment~\cite{wu_Experimentally_2017,wu_Quantum_2020,wu_Experimental_2021,xiong_Experimental_2021b}. 

In this work, we focus on two coherence manipulation protocols. One is the deterministic purification of coherence, which describes the process of extracting a pure coherent state from a general state via any possible free operations, such as the previously studied SIO~\cite{liu_Deterministic_2019}. Another coherence manipulation is the probabilistic enhancement protocol, which is the process that enhances the coherence of the quantum state via a stochastic strictly incoherent operation (sSIO)~\cite{liu_Enhancing_2017}. We prove two no-go theorems for the deterministic purification and probabilistic enhancement of coherence, respectively, which can be regarded as the complement of two coherence manipulation protocols.
More explicitly, we define a coherence measure and prove that a state cannot be deterministically purified via free operations when the state can be written as a convex combination of an incoherent state and a coherent state.
Besides, we give the sufficient and necessary conditions to distinguish whether the coherence of a quantum state can be probabilistically enhanced via sSIO based on the $l_1$-norm of coherence. 

The rest of this paper is organized as follows. In Sec.~\ref{sec:sectionII}, we present some preliminaries. In Sec.~\ref{sec:sectionIII}, we define a coherence measure and prove the no-go theorem for the deterministic purification of coherence.
In Sec.~\ref{sec:sectionIV}, we prove the no-go theorem for the probabilistic enhancement of coherence based on the $l_1$-norm of coherence. In Sec.~\ref{Sec:example}, we apply our results to the case of qubit states, and show these state that can neither be  deterministically purified nor probabilistically enhanced in the Bloch sphere. In Sec.~\ref{sec:sectionV}, we present our conclusions.

%take our results in the case of qubit states as an example 
%, which includes three principal issues: (i) the characterization, (ii) the quantification, and (iii) the manipulation 
%In the seminal work by Baumgratz \emph{et al.}, they proposed  . 
%The quantification includes proposing the essential conditions of coherence measures~\cite{baumgratz_quantifying_2014} and finding these legitimate measures~\cite{napoli_robustness_2016,yu_alternative_2016,liu_new_2017,qi_measuring_2017,zhou_polynomial_2017,xi_coherence_2019,xi_epsilon-smooth_2019,yao_quantum_2019,yu_quantifying_2020,li_quantum_2021}, which can be used to study the role of coherence in many physical contexts  quantitatively~\cite{bromley_frozen_2015,yu_measure-independent_2016,hillery_coherence_2016,shi_coherence_2017,streltsov_entanglement_2016,ma_coherence_2019,ma_operational_2019,yu_detecting_2019,ding_efficient_2021}.
%The manipulation refers to steer the coherence of a quantum state with limited free operation.
%the free states called
 
\section{Preliminaries}
\label{sec:sectionII}

\subsection{Resource Theory of Quantum Coherence}

In this work, we follow the framework of the resource theory of quantum coherence by Baumgratz \emph{et al.}~\cite{baumgratz_Quantifying_2014}. Let $\{\ket{i}\}_{i=0,1,...,d-1}$ be a prefixed basis in the $d$-dimensional Hilbert space. In the resource theory of quantum coherence, the density matrices that are diagonal in this specific basis are incoherent, being of the form $\sum_{i=0}^{d-1} p_i \ket{i}\bra{i}$ with probabilities $p_i$, and the set of incoherent states is denoted as $\mathcal{I}$. %incoherent states are those whose density matrices are diagonal in the basis,  is a probability distribution
The coherent states are those not of this form. 
%\red{and is denoted as $\mathcal{C}$}.
%An incoherent operation is required to fulfill $K_n\mathcal{I}K_n^\dagger\subset\mathcal{I}$ for $K_n$, i.e., each $K_n$ maps an incoherent state into an incoherent state. 
An IO~\cite{baumgratz_Quantifying_2014} is defined by a completely positive and trace preserving map, $\Lambda(\rho)=\sum_{n=1}^{N} K_{n} \rho K_{n}^{\dagger}$, where the Kraus operators fulfill not only $\sum_{n=1}^{N} K_{n}^{\dagger} K_{n}=I$ but also $K_{n} \mathcal{I} K_{n}^{\dagger} \subset \mathcal{I},$
i.e., each $K_{n}$ maps an incoherent state to an incoherent state.
%A strictly incoherent operation~\cite{winter_Operational_2016,yadin_Quantum_2016} is defined by a completely positive and trace preserving map, $\Lambda(\rho)=\sum_{n} K_{n} \rho K_{n}^{\dagger}$ with the Kraus operators fulfilling not only $\sum_{n=1}^{N} K_{n}^{\dagger} K_{n}=I$ and $K_{n} I K_{n}^{\dagger} \subset I$ but also $K_{n}^{\dagger} \mathcal{I} K_{n} \subset \mathcal{I},$ i.e., each $K_{n}$ as well $K_{n}^{\dagger}$ maps an incoherent state to an incoherent state.
An IO is called a SIO if each $K_n$ also satisfies $K_n^\dagger\mathcal{I}K_n\subset\mathcal{I}$~\cite{winter_Operational_2016,yadin_Quantum_2016}.

A functional $C$ can be taken as a coherence measure if it satisfies the four conditions~\cite{baumgratz_Quantifying_2014},

(C1) Non-negativity: $C(\rho)\geq 0$, and $C(\rho)=0$ if and only if $\rho \in \mathcal{I} $.

(C2) Monotonicity: C does not increase under the action of incoherent operations, i.e., $C (\rho) \geq C (\Lambda{ (\rho)}) $ if $\Lambda$ is  an incoherent operation.
%,a free operation in resource theory of coherence.

(C3) Strong monotonicity: $C$ does not increase on average under selective incoherent operations, i.e.,
\begin{eqnarray}
	\sum_{i} q_{i} C\left(\sigma_{i}\right) \leq C(\rho),
\end{eqnarray}
with probabilities $q_{i}=\operatorname{Tr}[K_{n} \rho K_{n}^{\dagger}]$, postmeasurement states $\sigma_{i}=K_{n} \rho K_{n}^{\dagger} / q_{i}$, and incoherent Kraus operators $K_{n}$.

(C4) Convexity: $C$ is a convex function of the state, i.e.,
\begin{eqnarray}
	\sum_{i} p_{i} C\left(\rho_{i}\right) \geq C\left(\sum_{i} p_{i} \rho_{i}\right) .
\end{eqnarray}

%that the sum of absolute values of its density matrix regarding the chosen basis 

The $l_1$-norm of coherence~\cite{baumgratz_Quantifying_2014} has been proven satisfying above four conditions, and its explicit expression is
\begin{eqnarray}
	C_{l_1}(\rho)=\sum_{i\neq j}\abs{\rho_{ij}}.
	\label{L}
\end{eqnarray}

\subsection{Probabilistic enhancement of coherence}
It is known that SIO does not increase the coherence of a state, i.e., $C(\Lambda(\rho))\leq C(\rho)$. However, when we only pick out those $\rho_n$ satisfying $C(\rho_n)>C(\rho)$ and discard other $\rho_n$ with smaller $C(\rho_n)$, we can probabilistically get a mixed state ${\sum_n}_{,C(\rho_n)>C(\rho)} p_n \rho_n$, which has a larger value of coherence. Therefore, we may enhance the coherence of a state via sSIO, denoted as $\Lambda_s(\rho)$, defined by
\begin{eqnarray}
	\Lambda_s(\rho)=\frac{\sum_{n=1}^L K_{n}\rho K_{n}^{\dagger}}{\Tr(\sum_{n=1}^LK_{n}\rho K_{n}^{\dagger})},
	\label{lams}
\end{eqnarray}
where $\{K_{1},K_{2},\dots, K_{L}\}$ is a subset of $N$ strictly incoherent Kraus operators of a SIO, and satisfies $\sum_{n=1}^L K_{n}^{\dagger}K_{n}\leq I$.
%which is constructed by a subset of strictly incoherent Kraus operators,, and $L \leq N$ 

For the sake of convenience, we recall some results from Ref.~\cite{liu_Enhancing_2017} without proofs.

For a state, $\rho=\sum_{ij}\rho_{ij}\ket{i}\bra{j}$, we further define three matrices $|\rho|$, $\rho_d$, and $\rho_d^{-\frac12}$, where $|\rho|=\sum_{ij}|\rho_{ij}|\ket{i}\bra{j}$, $\rho_d=\sum_i\rho_{ii}\ket{i}\bra{i}$, and $\rho_d^{-\frac12}$ is a diagonal matrix with elements
\begin{eqnarray}
	(\rho_d^{-\frac12})_{ii}=\left\{\begin{array}{ll} \rho_{ii}^{-\frac12}, &\text{if} ~ \rho_{ii}\neq0;\\
		0,&\text{if}~ \rho_{ii}= 0. \end{array}\right.
	\label{rho1/2}
\end{eqnarray}

With these notions, the maximal value of coherence that can be achieved for a state via a sSIO by means of the $l_1$-norm of coherence is expressed as~\cite{liu_Enhancing_2017},
\begin{eqnarray}
	\max_{\Lambda_s}C_{l_1}\left(\Lambda_s(\rho)\right)=\lambda_{\max}(\rho_d^{-\frac12}\abs{\rho}\rho_d^{-\frac12})-1,
	\label{theorem}
\end{eqnarray}
where $\lambda_{\max}(\rho_d^{-\frac12}\abs{\rho}\rho_d^{-\frac12})$ represents the maximum eigenvalue of the matrix $\rho_d^{-\frac12}\abs{\rho}\rho_d^{-\frac12}$.

\section{No-go theorem for the deterministic purification of coherence}
\label{sec:sectionIII}

We first propose a new functional and then prove it fulfills the above four conditions of the coherence measure. The definition is based on convex-roof construction. We define our measure for a pure state that sets zero for all pure incoherent states and one for all pure coherent states, named as Boolean-valued coherence measure, i.e.,
\begin{equation}
	C_{B}(\ket{\psi}) =  \begin{cases}
		1 , & \text{if} \quad  \ket{\psi}  \text{ is a coherent state}\\
		0 , & \text{if} \quad  \ket{\psi}  \text{ is an incoherent state}\\
	\end{cases}.
	\label{eq:trivalcohernce}
\end{equation}
%where $\mathcal{C}$ denotes the set of all pure coherent state.
 
We then extend our definition to the general case via convex-roof construction,
\begin{equation}
	C_{B}(\rho)=\inf_{\{p_{i},\ket{\psi_{i}}\}}\sum_{i=1}^{n} p_{i}C_{B}(\ket{\psi_{i}}),
	\label{eq:trivalconvexroof}
\end{equation}
where the infimum is taken over all the ensembles of $\{p_{i},\ket{\psi_{i}}\}$, i.e., $\rho = \sum_{i} p_{i}\ket{\psi_{i}}\bra{\psi_{i}} $.

Next, we prove that the functional $C_{B}$ is a legitimate coherence measure satisfied conditions (C1)-(C4).
\begin{proof}
	First, we show that the functional $C_{B}$ satisfies condition (C1). By definition, there is $C_{B} \geq 0$. Since an incoherent state $\delta$ admits a decomposition of the form $\delta = \sum_{i=0}^{d-1}p_{i} \ket{i}\bra{i}$, thus, $C_{B}(\delta) \leq \sum_{i} p_{i} C_{B}(\ket{i}) = 0 $. Hence, $C_{B}(\delta) = 0$ for an incoherent state $\delta$. Conversely, suppose that $C_{B}(\rho)=0$ for a state $\rho$. Then, there exists an ensemble $\left\{p_{n},\left|\varphi_{n}\right\rangle\right\}$ of $\rho$ such that $\sum_{n} p_{n} C_{B}\left(\left|\varphi_{n}\right\rangle\right)=0$, which further leads to $C_{B}\left(\left|\varphi_{n}\right\rangle\right)=0$ for all $n$. It follows that each $\left|\varphi_{n}\right\rangle$ is an incoherent state, and so is $\rho$.
	
	Second, we prove that the functional $C_{B}$ satisfies condition (C4), i.e., $C_{B}(\sum_{i}p_{i}\rho_{i}) \leq \sum_{i}p_{i} C_{B}(\rho_{i})$. For any $\rho_{i}$, we have a decomposition such as $\rho_{i}=\sum_{j}p_{j}^{i}|\psi_{j}^{i}\rangle \langle\psi_{j}^{i}|$ achieving the infimum in the definition of $ {C}_{B}\left(\rho_{i}\right)$.
	
	Thus,
	\begin{eqnarray}
		\begin{aligned}
			\sum_{i} p_{i} C_{B}\left(\rho_{i}\right) &=\sum_{i} p_{i} \sum_{j} p_{j}^{i}  {C}_{B}\left(\left|\psi_{j}^{i}\right\rangle\right) \\
			&=\sum_{i} \sum_{j} p_{i} p_{j}^{i}  {C}_{B}\left(\left|\psi_{j}^{i}\right\rangle\right) \\
			& \geq  {C}_{B}\left(\sum_{i} \sum_{j} p_{i} p_{j}^{i}\left|\psi_{j}^{i}\right\rangle\left\langle\psi_{j}^{i}\right| \right)\\
			&= {C}_{B}\left(\sum_{i} p_{i} \rho_{i}\right),
		\end{aligned}
		\label{eq:convex_coh}
	\end{eqnarray}
where the inequality follows from the definition of $C_{B}$.
	
	Third, we show that the functional $ {C}_{B}$ satisfies condition (C3).
	%, i.e., $C(\rho) \geq \sum_{n} p_{n} C\left(\rho_{n}\right)$, where $p_{n}=\operatorname{Tr}\left(K_{n} \rho K_{n}^{\dagger}\right), \rho_{n}=$
	%$K_{n} \rho K_{n}^{\dagger} / p_{n}$ for all $K_{n}$ with $\sum_{n} K_{n}^{\dagger} K_{n}=I$ and $K_{n} \mathcal{I} K_{n}^{\dagger} \subset \mathcal{I}$.
	
	Let $ {C}_{B}(\rho)=\sum_{i} q_{i} C_{B}\left(\left|\psi_{i}\right\rangle\left\langle\psi_{i}\right|\right)$ with $\rho=\sum_{i} q_{i}\left|\psi_{i}\right\rangle\left\langle\psi_{i}\right|$,
	\begin{eqnarray}
		\begin{aligned}
			\rho_{n} &=\frac{1}{p_{n}} K_{n} \rho K_{n}^{\dagger}=\frac{1}{p_{n}} K_{n} \sum_{i} q_{i}\left|\psi_{i}\right\rangle\left\langle\psi_{i}\right| K_{n}^{\dagger} \\
			&=\sum_{i} \frac{q_{i}}{p_{n}} K_{n}\left|\psi_{i}\right\rangle\left\langle\psi_{i}\right| K_{n}^{\dagger} \\
			&=\sum_{i} \frac{q_{i}}{p_{n}} \operatorname{Tr}\left(K_{n}\left|\psi_{i}\right\rangle\left\langle\psi_{i}\right| K_{n}^{\dagger}\right) \frac{K_{n}\left|\psi_{i}\right\rangle\left\langle\psi_{i}\right| K_{n}^{\dagger}}{\operatorname{Tr}\left(K_{n}\left|\psi_{i}\right\rangle\left\langle\psi_{i}\right| K_{n}^{\dagger}\right)}.
		\end{aligned}
	\end{eqnarray}
	
	Then, 
	\begin{eqnarray}
		\begin{aligned}
			 {C}_{B}(\rho) &=\sum_{i} q_{i} C_{B}\left(\left|\psi_{i}\right\rangle\left\langle\psi_{i}\right|\right) \\
			&=\sum_{i} q_{i} \sum_{n}\left\langle\psi_{i}\left|K_{n}^{\dagger} K_{n}\right| \psi_{i}\right\rangle C_{B}\left(\left|\psi_{i}\right\rangle\left\langle\psi_{i}\right|\right) \\
			& \geq \sum_{n} \sum_{i} q_{i}\left\langle\psi_{i}\left|K_{n}^{\dagger} K_{n}\right| \psi_{i}\right\rangle  {C}_{B}\left(\frac{K_{n}\left|\psi_{i}\right\rangle\left\langle\psi_{i}\right| K_{n}^{\dagger}}{\operatorname{Tr}\left(K_{n}\left|\psi_{i}\right\rangle\left\langle\psi_{i}\right| K_{n}^{\dagger}\right)}\right) \\
			&=\sum_{n} p_{n} \sum_{i} \frac{q_{i}}{p_{n}}\left\langle\psi_{i}\left|K_{n}^{\dagger} K_{n}\right| \psi_{i}\right\rangle  {C}_{B}\left(\frac{K_{n}\left|\psi_{i}\right\rangle\left\langle\psi_{i}\right| K_{n}^{\dagger}}{\operatorname{Tr}\left(K_{n}\left|\psi_{i}\right\rangle\left\langle\psi_{i}\right| K_{n}^{\dagger}\right)}\right) \\
			& \geq \sum_{n} p_{n}  {C}_{B}\left(\sum_{i} \frac{q_{i}}{p_{n}} \operatorname{Tr}\left(K_{n}\left|\psi_{i}\right\rangle\left\langle\psi_{i}\right| K_{n}^{\dagger}\right) \frac{K_{n}\left|\psi_{i}\right\rangle\left\langle\psi_{i}\right| K_{n}^{\dagger}}{\operatorname{Tr}\left(K_{n}\left|\psi_{i}\right\rangle\left\langle\psi_{i}\right| K_{n}^{\dagger}\right)}\right) \\
			&=\sum_{n} p_{n}  {C}_{B}\left(\rho_{n}\right).
		\end{aligned}
	\end{eqnarray}
	%\red{with the aid of the convexity of $ {C}_{B}$, we have}
The first inequality is based on the definition of ${C}_{B}$, and the second inequality is with the aid of the convexity of ${C}_{B}$.

Since (C3) and (C4) imply (C2), we obtain that the functional $C_{B}(\rho)$ also satisfies condition (C2). This completes the proof .
\end{proof}

By the way, a similar coherence measure, denoted as $C_{trivial}$, has been proposed in Ref.~\cite{peng_Maximally_2016}.  $C_{trivial}(\rho)$ is zero if $\rho$ is an incoherent state and $C_{trivial}(\rho)$ is one if $\rho$ is a coherent state. Here we should note that $C_{trivial}(\rho)$ does not satisfy condition (C4) for mixed states. To see this, suppose we have two quantum states,
\begin{equation}
	\rho_1=\left(
	\begin{array}{cc}
		1/2  & 1/2  \\
		1/2  & 1/2  
	\end{array}
	\right), \quad
	\rho_2=\left(
	\begin{array}{cc}
	1 & 0 \\0   & 0
	\end{array}
	\right) .
\end{equation}
The other state $\rho$ is defined as
\begin{equation}
	\rho=\frac{1}{2} \rho_{1} +\frac{1}{2} \rho_{2} =\left(
	\begin{array}{cc}
		3/4  & 1/4\\
		1/4  & 1/4  
	\end{array}
	\right).
\end{equation}
Note that $\rho_{1}$ and $\rho$ are the coherent states, while $\rho_2$ is an incoherent state. Then, we have 
\begin{eqnarray}
 C_{trivial}(\rho ) =1  > \frac{1}{2}= \frac{1}{2} C_{trivial}(\rho_1) + \frac{1}{2} C_{trivial}(\rho_2).
\end{eqnarray}
Thus,  $C_{trivial}$ does not satisfy condition (C4) for mixed states.

%A similar coherence measure has been proposed~\cite{peng_Maximally_2016} which is defined as a measure of the value zero if and only if its input state is incoherent; otherwise, it is always one. Note that $C_{trivial}$ and $C_{B}$ are in the same form for the pure state and not the same for the mixed state. In particular, we take a counterexample to show that $C_{trivial}$ does not satisfy condition (C4) for mixed states.

Based on the coherence measure $C_{B}$, we arrive at our first result.
\begin{theorem}
	The quantum state $\rho$ cannot be deterministically purified via any free operations in the resource theory of coherence if it can be written as a convex combination of an incoherent state and a coherent state.
	\label{Ob:1}
\end{theorem}

\begin{proof}[Proof of Theorem 1]
	Suppose there is a decomposition of the form, $\rho = \lambda \sigma + (1-\lambda) \tau$, where $\sigma$ is an incoherent state and $\tau$ is a coherent state.
	
	According to convexity of Eq.(\ref{eq:trivalconvexroof}), we have
	\begin{equation}
		C_{B}(\rho) \leq \lambda C_{B}(\sigma)+(1-\lambda)C_{B}(\tau)=(1-\lambda)C_{B}(\tau).
	\end{equation}
	Then, $C_{B} (\rho) $ must be less than 1, i.e., less than the coherence of all pure coherent states. According to condition (C2), no free operations in the resource theory of coherence can transform $\rho$ into a pure coherent state.
	
\end{proof}

The condition in Theorem \ref{Ob:1} can be transformed into a calculable function. Here, we introduce a functional called "coherent weight", $C_{W}{(\rho)}$~\cite{bu_Asymmetry_2018,yao_Anomalies_2020}, 
\begin{equation}
	C_{W}{(\rho)}=\min \{\gamma \geq 0: \rho=(1-\gamma) \sigma+ \gamma \tau, \sigma \in \mathcal{I}, \text { state } \tau \}
\end{equation}
and the corresponding numerical method is given in Ref.~\cite{bu_Asymmetry_2018}. The meaning of $C_{W}{(\rho)} = 1$ is that the state cannot be written as a convex combination of coherent state and incoherent state. %Therefore, we can infer a coherent mixed state can be purified via \red{the free operations in the resource theory of coherence} if $C_{W}{(\rho)} \neq 1$.

\begin{corollary}
	It is impossible to deterministically transform a full-rank mixed coherent state to a pure coherent state via any free operations in the resource theory of coherence.
	\label{Corollary2}
\end{corollary}	

\begin{proof}[Proof of Corollary 1]
	Suppose there is an $n$-dimensional full-rank mixed coherence state $\rho$, its form is $\rho=\sum_{i=1}^{n} \lambda_{i} \ket{\psi_{i}}\bra{\psi_{i}}$, thus $\lambda_{i} > 0$ for all $\lambda$. $\lambda_{min}$ is the minimum eigenvalue of the density matrix $\rho$. We construct a new matrix,
	\begin{equation}
		\delta = \rho - \lambda_{min}I,
	\end{equation}
	where $I$ denotes identity matrix. Thus, $\rho = \delta + \lambda_{min}I$.
	
	Normalizing $\delta$ and $\lambda_{min}I$, we obtain
	\begin{equation}
		\rho = (1- n \lambda_{min})(\frac{1}{1-n \lambda_{min}}\delta)+n \lambda_{min}(\frac{1}{n}I).
	\end{equation}
	Note that $n \lambda_{min}$ is not greater than 1 due to $\lambda_{min}$ is the minimum eigenvalue of the density matrix $\rho$. Then, $C_{B} (\rho) $ must be less than 1, i.e., less than the coherence of all pure coherent states. According to condition (C2), no free operations in the resource theory of coherence can transform a full-rank mixed coherent state into a pure coherent state.
	
\end{proof}

It is well known that if $\operatorname{supp}(\sigma) \subseteq \operatorname{supp}(\rho)$, then there is a decomposition of the form $\rho = \lambda \sigma + (1-\lambda) \tau$ with  $\tau$ is a density operator and $0 < \lambda \leq 1$~\cite{rudolph_Quantum_2004}, where $\operatorname{supp}(\rho)$ is the support of an Hermitian operator $\rho$, which is the vector space spanned by the eigenvectors of $\rho$ with non-zero eigenvalues. In Corollary \ref{Corollary2}, we consider a full-rank density matrix, so its support is the subspace of its eigenvectors. The maximally mixed state is incoherent in their support. Thus, it is impossible to deterministically purifying a full-rank mixed coherence state. The physical meaning of this result is that any full-rank state transformed from a pure state because of environmental noise will be an irremediable loss of coherence. We notice that this proof could be a simple alternative proof for the no-go theorem for quantum coherence deterministic purification in the case of a full-rank coherent state. [see~\cite{fang_Probabilistic_2018}, Theorem 3]. Meanwhile, our result is consistent with the coherent deterministic purification via SIO ~\cite{liu_Deterministic_2019}.

\section{No-go theorem for the probabilistic enhancement of coherence}
\label{sec:sectionIV}
It is worth noting that sSIO can increase the coherence of a state, i.e., $C(\Lambda_{s}(\rho))\geq C(\rho)$. 
%It is worth noting that the coherence of a state experiencing the sSIO is equal to or greater than before. 
Therefore, the no-go theorem for the probabilistic enhancement of coherence means that the coherence is the same as before without undergoing the sSIO, that is,
\begin{equation}
	\max_{{  \Lambda}_s}C_{l_1}\left({\Lambda}_s(\rho)\right) = {  \lambda}_{\max}(\rho_d^{-\frac12}\abs{\rho}\rho_d^{-\frac12})-1 = C_{l_1}(\rho).
	\label{condition}
\end{equation}

By solving the equation~(\ref{condition}), we can obtain the set of states that cannot be probabilistically enhanced, and the complement of the set is precisely the scope of application of the probabilistic enhancement of coherence. In the following,  we give the sufficient and necessary conditions to judge a state whether can be probabilistically enhanced.
%, requiring only standard arithmetic operations.

\begin{theorem}
	%Let $\rho \in \mathcal{H}_{d}$ be irreducible.
	The $l_1$-norm of coherence of a coherent state $\rho$ cannot be enhanced via sSIO if and only if $\rho_{ii}=\frac{\sum_{n \neq i}\left|\rho_{in}\right|}{C_{l_1}(\rho)}$, where $n= 0,1, \cdots, d-1$.
	\label{thm:SEirreducible}
\end{theorem}

\begin{proof}[Proof of Theorem 2]
	It is equivalent to prove that the maximum eigenvalue of $\rho_d^{-\frac12}\abs{\rho}\rho_d^{-\frac12}$ is $  C_{l_1}(\rho) + 1$ if and only if $\rho_{ii}=\frac{\sum_{n \neq i}\left|\rho_{in}\right|}{C_{l_1}(\rho)}$, where $n= 0,1,\cdots,d-1$.
	
	\emph{If part:} Assume $\rho$ satisfies $\rho_{ii}=\frac{\sum_{n \neq i}\left|\rho_{in}\right|}{C_{l_1}(\rho)}$, then we have
	\begin{eqnarray}
		\rho_{d}^{-\frac{1}{2}} \abs{\rho}\rho_{d}^{-\frac{1}{2}}=\sum_{ij}\frac{\left|\rho_{ij}\right|C_{l_1}(\rho)}{\sqrt{\sum_{n \neq i}\left|\rho_{in}\right|\sum_{m \neq j}\left|\rho_{jm}\right|}} \ket{i}\bra{j}.
	\end{eqnarray}
	
	%It is tedious that solving the eigenpair with the largest value of high-dimensions matrix.
	It is difficult to directly get the maximum eigenvalue and corresponding eigenvector for a high-dimensional density matrix.
	Thus, we adopt the method of verification to explain that $\ket{\varphi}=\sum_{i=1}^{n} \sqrt{\rho_{ii}}\ket{i}$ is an eigenvector of $\rho_{d}^{-\frac{1}{2}} \abs{\rho}  \rho_{d}^{-\frac{1}{2}}$.
	
	Substituting $\ket{\varphi}=\sum_{i=1}^{n} \sqrt{\rho_{ii}}\ket{i}$ into  $\rho_{d}^{-\frac{1}{2}} \abs{\rho}  \rho_{d}^{-\frac{1}{2}}\ket{\varphi} = {\lambda} \ket{\varphi}$, we have
	\begin{equation}
		\sum_{ij}\sqrt{\frac{C_{l_1}(\rho)}{\sum_{m \neq j} \abs{\rho_{jm}}}} \abs{\rho_{ij}}\ket{i}= {\lambda} (\sum_{i=1}^{n} \sqrt{\rho_{ii}}\ket{i}).
	\end{equation}
	Verifying $i$-th row,
	\begin{equation}
		\left(\sqrt{\rho_{ii}}+\sum_{i \neq j}{\abs{\rho_{ij}}}\sqrt{\frac{C_{l_1}(\rho)}{\sum_{i \neq j} \abs{\rho_{ij}}}}\right)\ket{i}= ({\lambda} \sqrt{\rho_{ii}})\ket{i}.
		\label{eq:18}
	\end{equation}
	Rewriting Eq. (\ref{eq:18}) as,
	\begin{equation}
		\sqrt{\rho_{ii}}+{C_l}_{1}(\rho)\sqrt{\frac{\sum_{i \neq j}\left|\rho_{ij}\right|}{C_{l_{1}}(\rho)}}={\lambda}(\sqrt{\rho_{ii}}).
		\label{eq:19}
	\end{equation}
	Due to $\rho_{ii}=\frac{\sum_{n \neq  i}\left|\rho_{in}\right|}{C_{l_1}(\rho)}$, we have ${\lambda} =C_{l_{1}}(\rho)+1$.
	
	Note that $\rho_{d}^{-\frac{1}{2}} \abs{\rho}  \rho_{d}^{-\frac{1}{2}}$  is a non-negative matrix, i.e., all elements are greater than or equal to 0, and we have verified $\ket{\varphi}=\sum_{i=1}^{n} \sqrt{\rho_{ii}}\ket{i}$ is an eigenvector of $\rho_{d}^{-\frac{1}{2}} \abs{\rho}  \rho_{d}^{-\frac{1}{2}}$, and the eigenvalue corresponding to this eigenvector is ${\lambda} =C_{l_{1}}(\rho)+1$. Thus, the spectral radius of $\rho_{d}^{-\frac{1}{2}} \abs{\rho}  \rho_{d}^{-\frac{1}{2}}$ is equal to $ C_{l_{1}}(\rho)+1$~\cite{horn_Matrix_2012}.
	Besides,  $\rho_{d}^{-\frac{1}{2}} \abs{\rho} \rho_{d}^{-\frac{1}{2}}$ is an Hermitian matrix, hence, the eigenvalues of $\rho_{d}^{-\frac{1}{2}} \abs{\rho} \rho_{d}^{-\frac{1}{2}}$ are real. Combining with the above results, we have ${\lambda}_{max} =C_{l_{1}}(\rho)+1 $.
	
	\emph{Only if part:} The matrix $\rho_d^{-\frac12}\abs{\rho}\rho_d^{-\frac12}$ is Hermitian, and the maximum eigenvalue is $C_{l_1}(\rho)+1$. For any normalized vector $\ket{\varphi}$, we have $\bra{\varphi}\rho_{d}^{-\frac{1}{2}} \abs{\rho}  \rho_{d}^{-\frac{1}{2}}\ket{\varphi} \leq \lambda_{max}(\rho_{d}^{-\frac{1}{2}} \abs{\rho} \rho_{d}^{-\frac{1}{2}})$, with equality if and only if $ \rho_{d}^{-\frac{1}{2}} \abs{\rho}  \rho_{d}^{-\frac{1}{2}}\ket{\varphi} =\lambda_{max} \ket{\varphi}$~\cite{horn_Matrix_2012}. Suppose there is a normalized vector $\ket{\varphi}=\sum_{i=1}^{n} \sqrt{\rho_{ii}}\ket{i}$, we can get
	\begin{equation}
		\begin{split}
			\bra{\varphi}\rho_{d}^{-\frac{1}{2}} \abs{\rho}  \rho_{d}^{-\frac{1}{2}}\ket{\varphi}&=\sum_{ijkl} \frac{\left| \rho_{ij} \right|}{\sqrt{\rho_{ii}\rho_{jj}}}\sqrt{\rho_{kk}} \sqrt{\rho_{ll}}  \bracket{l}{i} \bracket{j}{k} \\&= \sum_{ij} \abs{\rho_{ij}} = C_{l_1}(\rho) +1.
		\end{split}
	\end{equation}
	Thus, we can deduce that $\ket{\varphi}=\sum_{i=1}^{n} \sqrt{\rho_{ii}}\ket{i}$ is the normalized eigenvector corresponding to the maximum eigenvalue $C_{l_1}(\rho) + 1$.
	
	Substituting $\ket{\varphi}=\sum_{i=1}^{n} \sqrt{\rho_{ii}}\ket{i}$  into the expression $\rho_{d}^{-\frac{1}{2}} \abs{\rho}  \rho_{d}^{-\frac{1}{2}}\ket{\varphi}={  \lambda}\ket{\varphi} $, we have	

	\begin{equation}
		\sum_{ij}\frac{\left| \rho_{ij} \right|}{\sqrt{\rho_{ii}}}\ket{i}= {  \lambda}\sum_{i}\sqrt{\rho_{ii}}\ket{i}.
		\label{10}
	\end{equation}
	It follows that
	\begin{equation}
		{\lambda} = 1+\frac{\sum_{j \neq 1} \abs{\rho_{1j}}}{\rho_{11}}=1+\frac{\sum_{j \neq 2} \abs{\rho_{2j}}}{\rho_{22}}=\cdots=1+\frac{\sum_{j \neq n} \abs{\rho_{nj}}}{\rho_{nn}}.
		\label{13}
	\end{equation}
	Thus,
	\begin{equation}
		\rho_{22}=\rho_{11}\frac{\sum_{j \neq 2} \abs{\rho_{2j}}}{\sum_{j \neq 1} \abs{\rho_{1j}}} \quad , \cdots , \quad
		\rho_{nn}=\rho_{11}\frac{\sum_{j \neq n} \abs{\rho_{nj}}}{\sum_{j \neq 1} \abs{\rho_{1j}}}.
		\label{14}
	\end{equation}
	Due to $\sum_{i} \rho_{ii}= 1$, we can obtain, from Eq. (\ref{14}),
	\begin{equation}
		\rho_{ii}=\frac{\sum_{j \neq i}\left|\rho_{ij}\right|}{C_{l_1}(\rho)}.
		\label{rhoii}
	\end{equation}
\end{proof}

\section{Applications}
\label{Sec:example}

%\red{Here, we propose two no-go theorems for the state operation, specifically for the qubit case. These conclusions are direct results of our general theorems and meaningful for the practical qubit operations. } propose two no-go theorems in the case of qubit states

Here, we apply these above theorems to the qubit system. From Theorem \ref{Ob:1} and Corollary \ref{Corollary2}, for a qubit state, we conclude: the environment-induced coherence loss from a pure state is irremediable and can not be reconstructed via the free operations in the resource theory of coherence. This result is consistent with the result of qubit deterministic transformation via MIO, IO, and SIO~\cite{streltsov_Structure_2017}.

%any qubit state which is transformed from a pure state because of environmental noise will be an irremediable loss of coherence. In other words, the environment-induced coherence loss is forever and can not be reconstructed via \red{the free operations in the resource theory of coherence}. This result is consistent with the result of qubit deterministic transformation via MIO, IO, and SIO~\cite{streltsov_structure_2017}.

Following Theorem \ref{thm:SEirreducible}, for a qubit state, we have:

\begin{corollary}
	The form of the coherent qubit state that can never be enhanced under sSIOs is
	
	\begin{equation}
		\abs{\rho}=\left(
		\begin{array}{cc}
			1/2  & x \\
			x & 1/2\\
		\end{array}
		\right), 0 < x \leq 1/2.
		\label{singlecondition}
	\end{equation}
	
\end{corollary}

\begin{proof}[Proof of Corollary 2]
	Let us consider a qubit system. In the basis $\{|0\rangle,|1\rangle\}$, a quantum state can be generally expressed as
	\begin{eqnarray}
		\rho=\frac{1}{2}\left( \begin{array}{cc}1+r \cos \theta & e^{-i \varphi} r \sin \theta \\ e^{i \varphi} r \sin \theta & 1-r \cos \theta\end{array}\right),
		\label{eq:singlequbit}
	\end{eqnarray}
	where the parameters satisfy $0 < r \leq 1,0<\theta<\pi$, and $0 \leq \varphi \leq 2 \pi$ for coherent states.
	
	For this state, we have
	\begin{eqnarray}
		\rho_{d}^{-\frac{1}{2}}|\rho| \rho_{d}^{-\frac{1}{2}}=\frac{1}{2}\left( \begin{array}{cc}1 & \frac{r|\sin \theta|}{\sqrt{1-r^{2} \cos ^{2} \theta}} \\ \frac{r|\sin \theta|}{\sqrt{1-r^{2} \cos ^{2} \theta}} & 1\end{array}\right).
	\end{eqnarray}
	The qubit state $\rho$ can never be enhanced under sSIOs mean that $\max_{{  \Lambda}_s}C_{l_1}\left({\Lambda}_s(\rho)\right)=C_{l_{1}}(\rho)$, i.e.,
	\begin{equation}
		{ \lambda}_{\max}(\rho_d^{-\frac12}\abs{\rho}\rho_d^{-\frac12})-1=\frac{r|\sin \theta|}{\sqrt{1-r^{2} \cos ^{2} \theta}}=r|\sin \theta|.
	\end{equation}
	This means that $\cos \theta$ is 0.
	
	Substituting $\cos \theta=0$ into the equation (\ref{eq:singlequbit}), we have
	\begin{eqnarray}
		\rho=\frac{1}{2}\left( \begin{array}{cc}1 & e^{-i \varphi} r \\ e^{i \varphi} r  & 1\end{array}\right),
		\label{eq:singlequbit_2}
	\end{eqnarray}
	which is equivalent to Eq.(\ref{singlecondition}).
\end{proof}

We visualize the states that can neither be deterministically purified nor probabilistically enhanced in the Bloch sphere, as shown in Fig.~\ref{fig:1}. Geometrically, these states lie on the equatorial plane of the Bloch sphere. 

\begin{figure}
	\centering
	\includegraphics[width=0.5\columnwidth]{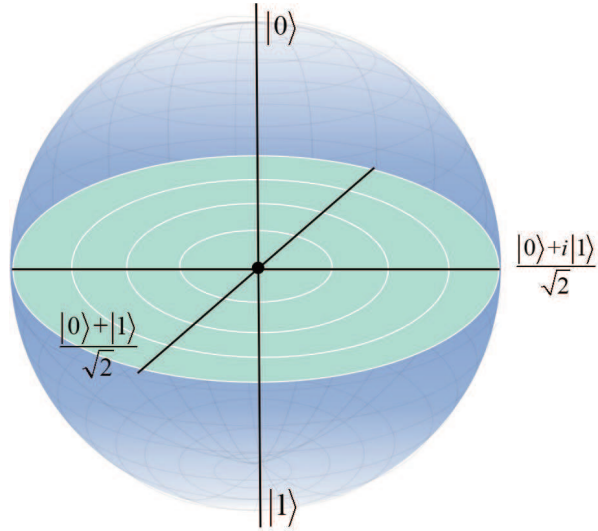}
	\caption{Schematic plot of qubit states that can neither be deterministically purified nor probabilistically enhanced. Geometrically, these states are in green at the Bloch sphere.}    
	\label{fig:1}
\end{figure}

\section{Conclusions}
\label{sec:sectionV}

In summary, we have provided two no-go theorems for the deterministic purification and probabilistic enhancement of coherence, respectively. More explicitly, we have shown that any free operations in the resource theory of coherence cannot complete the deterministic purification of coherence when the aim state can be written as a convex combination of an incoherent state and a coherent state. Besides, we have presented a simple condition to distinguish whether a state can be probabilistically enhanced via sSIO. 

Our results can be regarded as the complement for the deterministic purification and probabilistic enhancement of coherence and provide a theoretical criterion for experimental deterministic purification and probabilistic enhancement of coherence. Furthermore, the no-go theorem for deterministic purification of coherence may be further extended to any convex resource theory, which also provides a simple proof of no-go theorems for quantum resource deterministic purification~\cite{fang_NoGo_2020}.

%i.e., for any convex resource theory, if a state can be written as a convex combination of a resource state and a non-resource state, then the state cannot be deterministically purified via \red{the free operations in the resource theory of coherence}.

\ack
We thank C. L. Liu for thoroughly reading the manuscript, and for many suggestions, corrections, and comments, which have certainly helped to improve this manuscript. We are grateful to two anonymous referees for providing very useful comments on earlier versions of this manuscript. This work was supported by NSF China through Grant No. 11575101.

\bibliographystyle{iopart-num}
%\bibliography{newcite}
%\bibliography{Cite}
\providecommand{\newblock}{}

\end{document}